\documentclass[ reprint,
superscriptaddress,
 amsmath,amssymb,
 aps,
]{revtex4-1}

\usepackage{subfigure}

\usepackage{hyperref}

\usepackage{xcolor}
\usepackage{tikz}
\usetikzlibrary{quantikz}

\usepackage{algpseudocode}

\usepackage{graphicx}
\usepackage{dcolumn}
\usepackage{bm}

\usepackage[normalem]{ulem} 

\usepackage{amsmath}
\usepackage{amssymb}
\usepackage{amsthm}
\usepackage[utf8]{inputenc}
\usepackage[english]{babel}
\usepackage{mathtools}

\newtheorem{theorem}{Theorem}
\newtheorem{proposition}{Proposition}

\begin{document}

\title{Subspace Variational Quantum Simulation: Fidelity Lower Bounds as Measures of Training Success}

\author{Seung Park}
\affiliation{Institute for Convergence Research and Education in Advanced Technology, Yonsei University, Seoul 03722, Republic of Korea}
\affiliation{Department of Quantum Information, Graduate School, Yonsei University, Seoul 03722, Republic of Korea}
\author{Dongkeun Lee}
\affiliation{Center for Quantum Information R\&D, Korea Institute of Science and Technology Information, Daejeon 34141, Republic of Korea}
\author{Jeongho Bang}
\affiliation{Institute for Convergence Research and Education in Advanced Technology, Yonsei University, Seoul 03722, Republic of Korea}
\affiliation{Department of Quantum Information, Graduate School, Yonsei University, Seoul 03722, Republic of Korea}
\author{Hoon Ryu}
\affiliation{School of Computer Engineering, Kumoh National Institute of Technology, Gumi, Gyeongsangbuk-do 39177, Republic of Korea}
\author{Kyunghyun Baek}
\email{k.baek@yonsei.ac.kr}
\affiliation{Institute for Convergence Research and Education in Advanced Technology, Yonsei University, Seoul 03722, Republic of Korea}\
\affiliation{Department of Quantum Information, Graduate School, Yonsei University, Seoul 03722, Republic of Korea}
\date{\today}

\begin{abstract}
We propose an iterative variational quantum algorithm to simulate the time evolution of arbitrary initial states within a given subspace. The algorithm compresses the Trotter circuit into a shorter-depth parameterized circuit, which is optimized simultaneously over multiple initial states in a single training process using fidelity-based cost functions. After the whole training procedure, we provide an efficiently computable lower bound on the fidelities for arbitrary states within the subspace, which guarantees the performance of the algorithm in the worst-case training scenario. We also show our cost function exhibits a barren-plateau-free region near the initial parameters at each iteration in the training landscape. The experimental demonstration of the algorithm is presented through the simulation of a 2-qubit Ising model on the quantum processor. As a demonstration for a larger system, a simulation of a 10-qubit Ising model is also provided.
\end{abstract}

\maketitle

\section{Introduction}

The simulation of quantum systems is one of the promising applications of quantum computers by comprehending their dynamic properties from scientific perspectives \cite{Feynman1982, daley2022practical, cirac2012goals, altman2021quantum}. Furthermore, simulating quantum systems has served as a foundational method for deriving other useful quantum algorithms \cite{abrams1999quantum, ding2023even}.
Various methods have been developed to simulate quantum systems efficiently encompassing techniques such as product formula \cite{Lloyd1996, Childs2021, layden2022first}, the linear combination of unitary operations \cite{childs2012hamiltonian}, quantum signal processing \cite{low2017optimal}, and qubitization \cite{low2019hamiltonian}. In particular, the product formula-based algorithm has undergone refinements aimed at enhancing its error limits \cite{Childs2021, layden2022first} for the regime where constant factors are more important than asymptotic scaling due to large errors. As a result, the simulation of quantum systems is widely recognized as offering a practical quantum advantage in near-term quantum computing. \cite{georgescu2014quantum, daley2022practical}. The implementation of quantum simulation, nevertheless, is still challenging to achieve the quantum advantage even with the efforts to reduce the computational resources \cite{campbell2021early}.

Variational quantum simulation (VQS) \cite{Li2017} was introduced in the framework of variational quantum algorithms (VQAs) \cite{Cerezo2021_vqa,Bharti2022_noisy, Benedetti2021, commeau2020variational} to alleviate the issues posed by noise and limited circuit depth in implementing quantum simulation algorithms. In particular, fidelity-based iterative VQAs for quantum simulation have attracted significant attention due to their noise-resilience feature, which enables executing the algorithm within timescales shorter than the limited coherence time of the device \cite{otten2019noise}. These methods compress the circuit depth required to simulate the dynamics of quantum systems---especially in approaches based on product formulas---by leveraging the entangling power of parameterized quantum circuits (PQCs) \cite{berthusen2022quantum, lin2021real, Matteo2025, barison2021efficient}. Moreover, these approaches have shown promise in addressing the barren plateau problem \cite{haug2021opt}---that leads to exponential sampling costs for gradient estimation---and also in exhibiting approximately convex training landscapes \cite{puig2025variational}.

Quantum simulation algorithms generally enable the simulation of quantum state time evolution independently of the choice of initial state. In contrast, iterative variational approaches focus on the dynamics of a specific initial state. However, it is often necessary to simulate  a subset of states within a low-energy subspace of the Hamiltonian \cite{Sahinoglu2021, gong2024complexity, zlokapa2024hamiltonian}. For this purpose, the simulation of a set of states within a subspace, the so-called subspace variational quantum simulation, has been developed \cite{heya2023subspace}. However, this approach requires executing a subspace-search variational quantum eigensolver \cite{Nakanishi2019_ssvqe}, which is generally QMA-complete  \cite{Kempe2006}---and thus not expected to be efficiently solvable even by quantum computers---and also vulnerable to the barren plateau phenomenon.

In this work, we propose an algorithm that captures the Trotter time evolution of multiple states living in a chosen subspace using PQCs. The algorithm trains the PQC to iteratively follow the Trotter time evolution within the subspace at each time step, thereby compressing the overall dynamics into a shorter-depth PQC within the VQA framework, for a subspace that is expressible by the chosen ansatz. By optimizing the PQC simultaneously over the Trotter evolution of $d$ orthonormal basis states of the subspace and $d-1$ specific linear combinations thereof, the trained PQC can approximately reproduce the time evolution of arbitrary states within the subspace. 

We assess the success of training through the fidelity between the state evolved under the Trotter circuit and the parameterized state, as is common in studies of iterative variational approaches \cite{berthusen2022quantum, lin2021real}. In this context, a lower bound on the fidelity serves as a performance guarantee for the algorithm in worst-case training scenarios, whereas the performance of VQAs is generally difficult to assess quantitatively.
We introduce such a bound that can be efficiently computed via semidefinite programming, using constraints based on the fidelities of the basis states that are naturally obtained from the training history. We also show that training the PQC for the Trotter time evolution of a given subspace can be carried out in a ‘warm-start’ regime \cite{puig2025variational}, meaning the training landscape includes a region free from barren plateaus. 

In the following, we first present the main structure of our algorithm---covering the cost function, the evaluation of fidelity lower bounds, and the validation of the warm-start regime in the training process---in Sections~\ref{sec:framework} to~\ref{sec:ws}.  We then provide experimental and simulation results. 
Section~\ref{sec:Ising} describes the 2-qubit experiments conducted on the quantum processor, demonstrating our algorithm in practice. Section~\ref{sec:10Ising} presents simulations of a 10-qubit system as an example of its implementation in a larger system.

\section{Methods}\label{sec:algorithm}

\subsection{Framework of subspace time evolution in parameterized quantum circuits}

\label{sec:framework}

\begin{figure*}[t]
\centering
  \includegraphics[width=16cm]{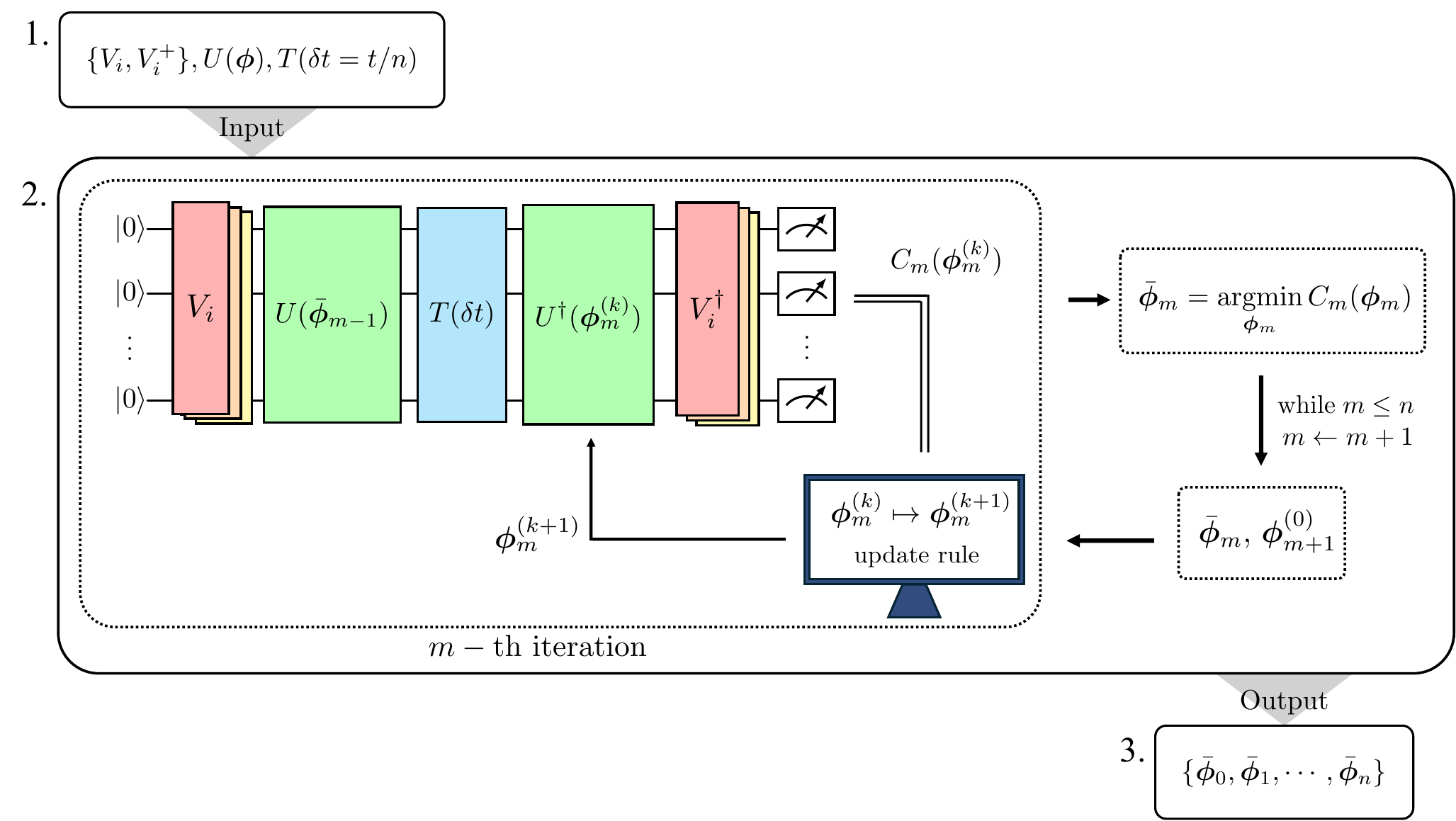}

\caption{\label{fig:scheme} Schematic diagram of the algorithm. Each numbered box corresponds to the matching step in the algorithm described above.
1. Preparation of the initial states.
2. The process of the $m$-th iteration. The cost function is evaluated by the quantum circuit shown, via fidelity measurements. The folded gates ${V_i}$ prepare the initial states such that $|\Psi_i\rangle = V_i|\bm{0}\rangle$ (similarly for $|\Psi_i^+\rangle$). Based on the cost function value, the classical computer (navy monitor) updates the variational parameters according to a chosen update rule (e.g., gradient descent), and the cost function is minimized iteratively.
3. After completing the iteration for $m = 1$ to $n$, the optimized parameters ${\bar{\bm{\phi}}_m}$ are returned as the output.}
\end{figure*}

Trotter approximation is a standard technique widely used in quantum simulation \cite{Lloyd1996, georgescu2014quantum}, which can efficiently simulate the time evolution of a system described by a $k$-local Hamiltonian. Specifically, for a given Hamiltonian $H = \sum_{l=1}^L H_l$, where each local term $H_l$ involves at most $k$ interactions, the dynamics of the system with an initial state $|\Psi_0\rangle$ can be simulated as
\begin{equation}
e^{-iHt}|\Psi_0\rangle \sim T(\delta t)^n|\Psi_0\rangle
\end{equation}
for sufficiently large $n$ determined based on the desired accuracy, where $T(\delta t)$ is Trotter circuit for the time slice $\delta t=t/n$.
However, the real-world implementation of the algorithm remains challenging due to limitations of current quantum computers, such as noise and limited circuit depth. One approach to dealing with these challenges is to utilize a VQA that leverages the entangling power of PQCs by compressing the Trotter circuit \cite{lin2021real}.

We develop a VQA that aims to capture the Trotter time evolution of a subspace spanned by chosen orthonormal initial states using a PQC. The algorithm consists of: a set of orthonormal initial states $\{|\Psi_i\rangle\}_{i=0}^{d-1}$ spanning the subspace $\mathcal{S}$; a PQC ansatz $U(\bm{\phi})$; and the Trotter circuit $T(\delta t)$ for the time step $\delta t = t/n$, where $t$ is the total simulation time. The PQC ansatz iteratively trains a target unitary 
\begin{align}\label{eq:U_m}
    \mathcal{U}_m\coloneqq \sum_{i=0}^{d-1} T(\delta t)^m|\Psi_i\rangle \langle \Psi_i| + \mathcal{U}_\perp,
\end{align}
by increasing the time step $m=0,1,\cdots, n$. Here $\mathcal{U}_\perp$ does not affect the evolution of states in $\mathcal S$, thus can be arbitrarily chosen. 

In detail, we define a cost function for $m$-th iteration of the algorithm as
\begin{widetext}
\begin{align}
    C_m(\bm{\phi}_m) 
    \coloneqq 1-\frac{1}{2d-1}\left(\sum_{i=0}^{d-1} \left|\langle\Psi_i(\bm{\phi}_m)|T(\delta t) | \Psi_i(\bar{\bm{\phi}}_{m-1})\rangle\right|^2+\sum_{i=1}^{d-1} \left| \langle\Psi_i^+(\bm{\phi}_m)|T(\delta t)|\Psi_i^+(\bar{\bm{\phi}}_{m-1})\rangle\right|^2\right),\label{eq:cost_m}
\end{align}
\end{widetext}
with $|\Psi_i^+\rangle\coloneqq (|\Psi_0\rangle+|\Psi_i\rangle)/\sqrt{2}$ and $|\Psi(\bm{\phi})\rangle\coloneqq U(\bm{\phi})|\Psi\rangle$. The fixed parameter $\bar{\bm{\phi}}_{m-1}$ is optimal parameter obtained from $(m-1)$-th training step, such that $|\Psi_i(\bar{\bm{\phi}}_{m-1})\rangle\approx T(\delta t)^{m-1}|\Psi_i\rangle$. The cost function is designed to reproduce the target unitary $\mathcal{U}_m$ using the PQC, $U(\bar{\bm{\phi}}_m)$, at each iteration. To this end, optimizing the first term, sum of the fidelities of the basis states in the cost function, ensures that $U(\bar{\bm{\phi}}_m)|\Psi_i\rangle \approx T^m(\delta t)|\Psi_i\rangle$ for all $i$, up to a global phase. However, when optimizing over multiple states simultaneously, each state may acquire a different global phase, resulting in uncontrolled relative phases between the basis states for different $i$. Regulating these relative phases is crucial for accurately reproducing the sub-block of the Trotter matrix (matrix representation of Trotter circuit in the basis $\{|\Psi_i\rangle\}$) within the chosen subspace---an issue that does not arise when targeting the time evolution of a single initial state. Since fidelity is defined as the squared magnitude of the overlap, relative phase information is inherently lost. Therefore, the second term of the cost function---based on fidelities of pairwise superpositions of the basis states---is necessary to recover information about the relative phases. The necessity of including the states $\{|\Psi_i^+\rangle\}$ is demonstrated in Appendix~\ref{sec:comp_fewer}.

Based on this cost function, our algorithm proceeds as follows:
\begin{enumerate}
\item Prepare the following as input: the orthonormal initial states $\{|\Psi_i\rangle\}_{i=0}^{d-1}$ along with their linear combinations $\{|\Psi_i^+\rangle\}_{i=1}^{d-1}$; the PQC $U(\bm{\phi})$; and the Trotter circuit $T(\delta t)$ for a single time step $\delta t = t/n$, where $t$ is the total simulation time and $n$ is the number of time steps.
\item Set $m=0$, and initialize $\bar{\bm{\phi}}_0$ such that $U(\bar{\bm{\phi}}_0) = I$. For $m=1,2,\cdots, n$, iteratively optimize the cost function as: 
\begin{equation}
\bar{\bm{\phi}}_m = \underset{\bm{\phi}_m}{\textrm{argmin}}\,C_m(\bm{\phi}_m),
\end{equation} 
\item The output is $\{\bar{\bm{\phi}}_0,\bar{\bm{\phi}}_1,\cdots,\bar{\bm{\phi}}_n\} $.
\end{enumerate}
Figure~\ref{fig:scheme} is a schematic diagram of our algorithm. The quantum circuit used to evaluate the cost function by measuring fidelity \cite{rethinasamy2023estimating} at each iteration is also shown in the figure.

\subsection{Lower bound of fidelity}
\label{sec:lb}

In our framework, the training success is assessed through the infidelity between the target state, evolved under the Trotter circuit, and the parameterized state. Thus, it is natural to evaluate the lower (upper) bound of the fidelity (infidelity) as a performance guarantee for the worst-case training scenario.

As described in Eq.~\eqref{eq:cost_m}, the cost function is designed for $|\Psi_i(\bm{\phi}_m)\rangle$  to approximately simulate $T(\delta t)U(\bar{\bm{\phi}}_{m-1})|\Psi_i\rangle$, rather than  $T(\delta t)^m|\Psi_i\rangle$, thereby maintaining the advantage of shorter-circuit-depth implementation across iterations. However, our primary interest lies in simulating $T(\delta t)^m|\Psi_i\rangle$. Thus, to assess the quality of training procedure, we derive a lower bound on the fidelities between the optimized state $|\Psi_i(\bar{\bm{\phi}}_m)\rangle$ and the state evolved under the $m$-fold Trotter circuit  $T(\delta t)^m|\Psi_i\rangle$, as stated in the following proposition, which will be proven in Appendix~\ref{sec:proof_lb}:

\begin{proposition}\label{prop:lb_m}
The fidelity between $|\Psi_i(\bar{\bm{\phi}}_m)\rangle$ and the ideally evolved state $T(\delta t)^m|\Psi_i\rangle$ under $m$-fold Trotter circuit is lower bounded by
\begin{align}\label{eq:lowerbound_m}
\cos^2 \left(\min\left[\sum_{j=1}^m \arccos\left(\sqrt{f_{i,j}}\right), \frac{\pi}{2}\right]\right),
\end{align}
where $f_{i,j}$ is the fidelity obtained at $j$-th iteration, defined as $f_{i,j} \coloneqq |\langle\Psi_i(\bar{\bm{\phi}}_j)|T(\delta t)|\Psi_i(\bar{\bm{\phi}}_{j-1})\rangle|^2$. 
\end{proposition}
\noindent 
During the training of PQCs, the fidelity between $|\Psi_i(\bar{\bm{\phi}}_m)\rangle$ and $T(\delta t)|\Psi_i(\bar{\bm{\phi}}_{m-1})\rangle$ at each step $m$ is naturally obtained. 
Consequently, one can estimate how well the PQCs approximate the $m$-fold Trotter circuit $T(\delta t)^m$ (restricted to the subspace) without directly evaluating the fidelity $|\langle\Psi_i(\bar{\bm{\phi}}_m)|T(\delta t)^m|\Psi_i\rangle|^2$, at least in terms of a lower bound. Furthermore, these lower bounds serve as constraints for computing the fidelity lower bound of arbitrary states within the subspace, as described below.

The ultimate goal of our method is to approximate the time evolution of a subspace, rather than that of individual states. Consequently, it is essential to establish a lower bound on the fidelity between the Trotter evolution and its approximation for an \emph{arbitrary initial state} within the given subspace. Specifically, given sets of $d$ orthonormal approximated states (the parameterized states) ${|\Psi_0\rangle, \cdots, |\Psi_{d-1}\rangle}$ and the corresponding target states (the states evolved under the Trotter circuit) ${|\Phi_0\rangle, \cdots, |\Phi_{d-1}\rangle}$, we aim to evaluate the minimum of $|\langle \Phi | \Psi \rangle|^2$, where $|\Phi\rangle = \sum_{i=0}^{d-1} c_i |\Phi_i\rangle$ and $|\Psi\rangle = \sum_{i=0}^{d-1} c_i |\Psi_i\rangle$. According to Proposition~\ref{prop:lb_m}, the constraints of this minimization are derived from experimentally obtained fidelities $F_i$ and $F_i^+$, expressed as $F_i \leq |\langle \Phi_i | \Psi_i \rangle|^2 \leq 1$ and $F_i^+ \leq |\langle \Phi_i^+ | \Psi_i^+ \rangle|^2 \leq 1$, respectively, where $|\Phi_i^+\rangle = (|\Phi_0\rangle + |\Phi_i\rangle)/\sqrt{2}$ for $i > 0$, and similarly for $|\Psi_i^+\rangle$. This optimization problem is formulated as a quadratically constrained quadratic program (QCQP):
\begin{align}
\min&\; \bm{f}^\dagger C \bm{f}\\\textrm{s.t.}&\;F_\alpha \leq \bm{f}^\dagger C_\alpha\bm{f}\leq 1,\;\;\alpha=0,1,\cdots,d-1,\nonumber\\&\; F_\beta^+\leq \bm{f}^\dagger C_\beta^+ \bm{f}\leq 1,\;\; \beta=1,2,\cdots,d-1,\nonumber\\&\;\phantom{F_\gamma\leq}\;\;\bm{f}^\dagger C^{\text{Nor}}_\gamma\bm{f}\leq 1,\;\; \gamma=0,1,\cdots,d-1,\nonumber
\end{align}
where $\bm{f}=(\langle\Phi_0|\Psi_0\rangle,\langle\Phi_0|\Psi_1\rangle,\cdots,\langle\Phi_{d-1}|\Psi_{d-1}\rangle)$. For convenience, we employ $d$-based digit indexing for the vector and matrix elements, ranging from $00$ to $(d-1)(d-1)$, such that $[\bm{f}]_{ij} = \langle\Phi_i|\Psi_j\rangle$. 
The coefficient matrix $C$ is Hermitian and positive semidefinite, with elements defined as $[C]_{ij,kq} = c_i^*c_j c_kc_q^*$. The constraint matrices are explicitly defined as:
\begin{align}
[C_\alpha]_{ij,kq} &=
\begin{cases}
1 & \text{if } i = j = k = q = \alpha \\
0 & \text{otherwise}
\end{cases} \\
[C_\beta^+]_{ij,kq} &=
\begin{cases}
1/4 & \text{if } ij, kq \in \{00, 0\beta, \beta0, \beta\beta\} \\
0 & \text{otherwise}
\end{cases} \\
[C^{\text{Nor}}_\gamma]_{ij,kq} &=
\begin{cases}
\delta_{ij, kq} & \text{if } j = q = \gamma \\
0 & \text{otherwise}
\end{cases},
\end{align}
where $C_\alpha$, $C_\beta^+$ impose the constraints corresponding to the fidelities $F_\alpha$, $F_\beta^+$, while $C_\gamma^{\text{Nor}}$ stands for the condition that the fidelity of any state is less than or equal to 1.
However, it is important to note that this QCQP is generally non-convex and NP-hard \cite{xu2023new}, implying that finding the global minimum is infeasible for large-scale instances.

The problem can be reformulated as semidefinite programming (SDP) by introducing a matrix $F\coloneqq \bm{f}\bm{f}^\dagger$, where $F\in \mathbb{C}^{d^2\times d^2}$, $F\succeq0$ and $\text{Rank}[F]=1$. However, this formulation remains equivalent to the original QCQP and therefore still not solvable in polynomial time. Since the rank constraint poses the primary computational difficulty, we relax it to obtain a standard SDP:
\begin{align}\label{eq:sdr}
\min&\; \text{Tr}\left[F C\right]\\\textrm{s.t.}&\;F\succeq 0,\nonumber\\&\;F_\alpha \leq \text{Tr}\left[F C_\alpha\right]\leq 1,\;\;\alpha=0,1,\cdots,d-1,\nonumber\\&\; F_\beta^+\leq \text{Tr}\left[F C_\beta^+\right]\leq 1,\;\; \beta=1,2,\cdots,d-1,\nonumber\\&\;\phantom{F_\gamma\leq}\;\;\text{Tr}\left[F C^{\text{Nor}}_\gamma\right]\leq 1,\;\; \gamma=0,1,\cdots,d-1\nonumber.
\end{align}
The solution to this relaxed problem provides a weaker lower bound on the fidelity. This approach, known as semidefinite relaxation (SDR) \cite{goemans1995improved}, can be solved efficiently, and its solution is guaranteed to be globally optimal.
\begin{theorem}\label{thm:lb}
The SDP in Eq.~\eqref{eq:sdr} can be solved efficiently using resources that scale polynomially with the dimension $d$. Moreover, the problem exhibits strong duality when $F_\alpha, F_\beta^+ < 1$ for all $\alpha, \beta$, a condition that is trivially satisfied in practical experiments.
\end{theorem}
The proof is given in Appendix~\ref{sec:proof_lb}. 

Note that if $F_\alpha=F_\beta^+=1$ for all $\alpha,\beta$, the PQC perfectly reproduces the sub-blocks of the Trotter matrix 
(restricted to the subspace) by training on $\{|\Psi_i\rangle, |\Psi_i^+\rangle\}$, and the fidelity of every arbitrary state within the subspace reaches $1$. 
Any deviation from this ideal case stems from imperfect training, where the fidelities fall below one. In such cases, the lower bound offers the worst-case scenario---the least favorable choice of parameters consistent with the constraints $\{F_0,F_1,F_1^+\}$.

\begin{figure}[t]
        \centering

        \subfigure[]{\includegraphics[width=4.2cm]{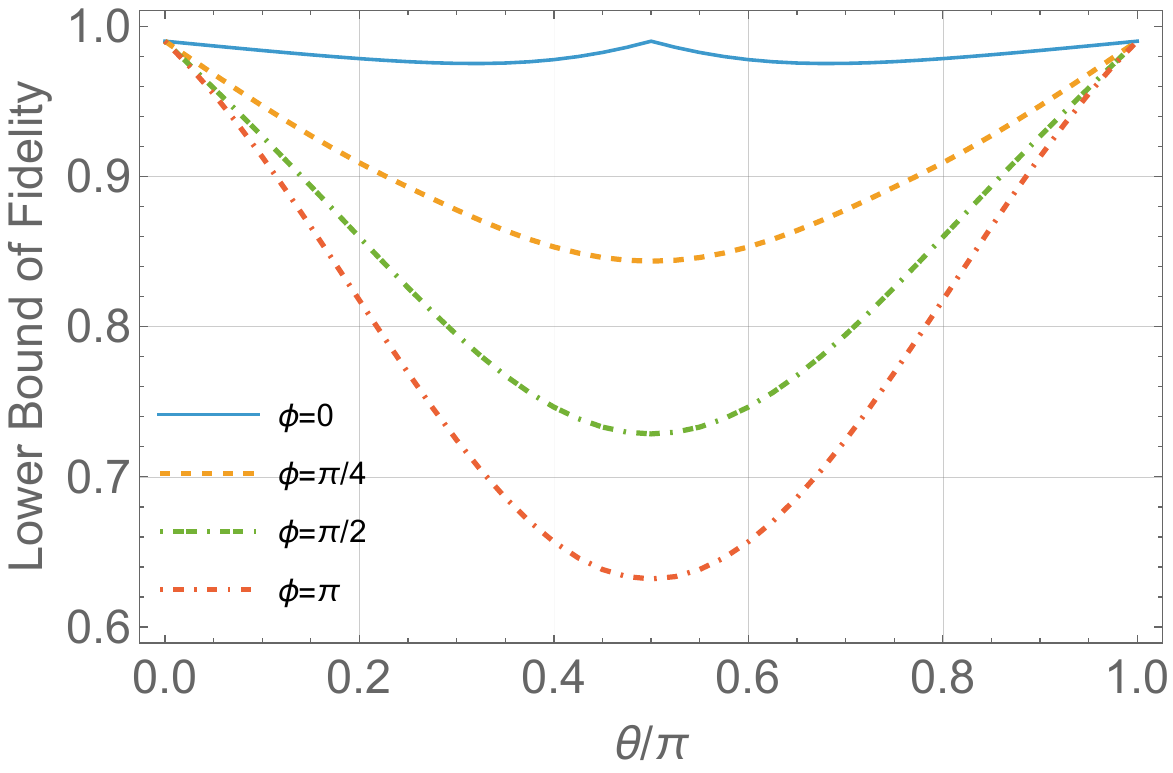}\label{fig:fids_lb}}
        \subfigure[]{\includegraphics[width=4.2cm]{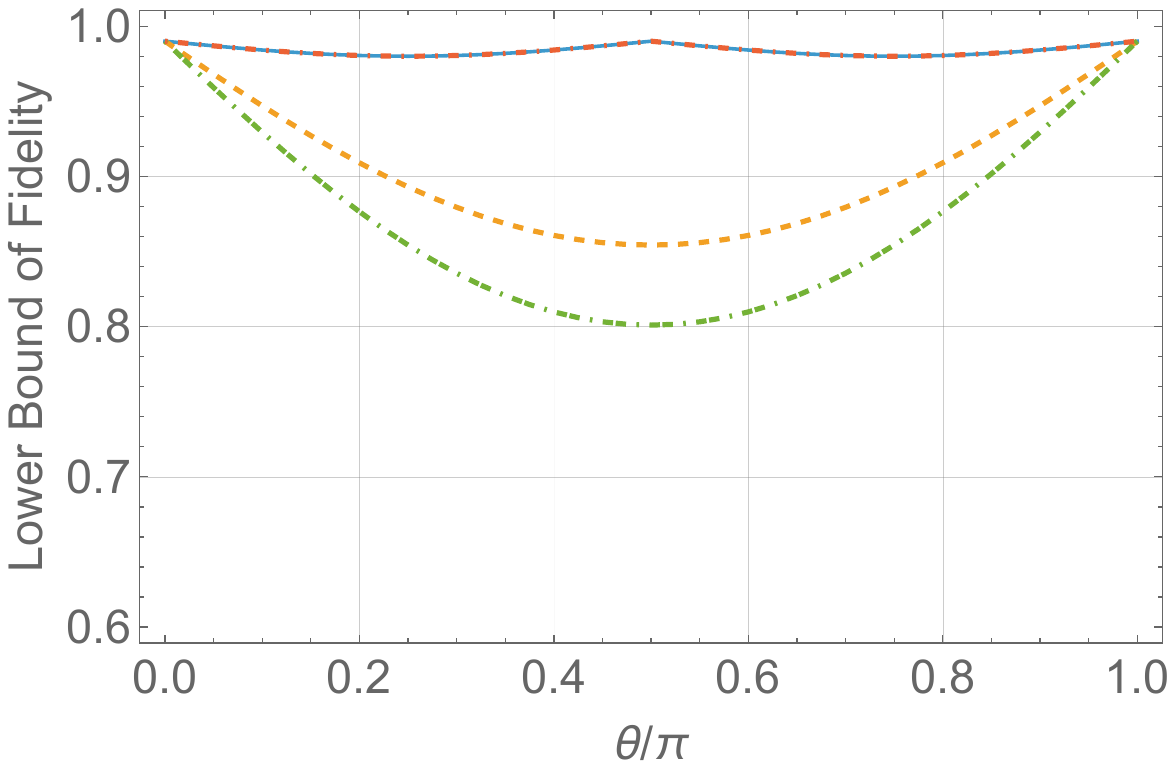} \label{fig:fids_lb_m}  }
    \subfigure[]{\includegraphics[width=4.2cm]{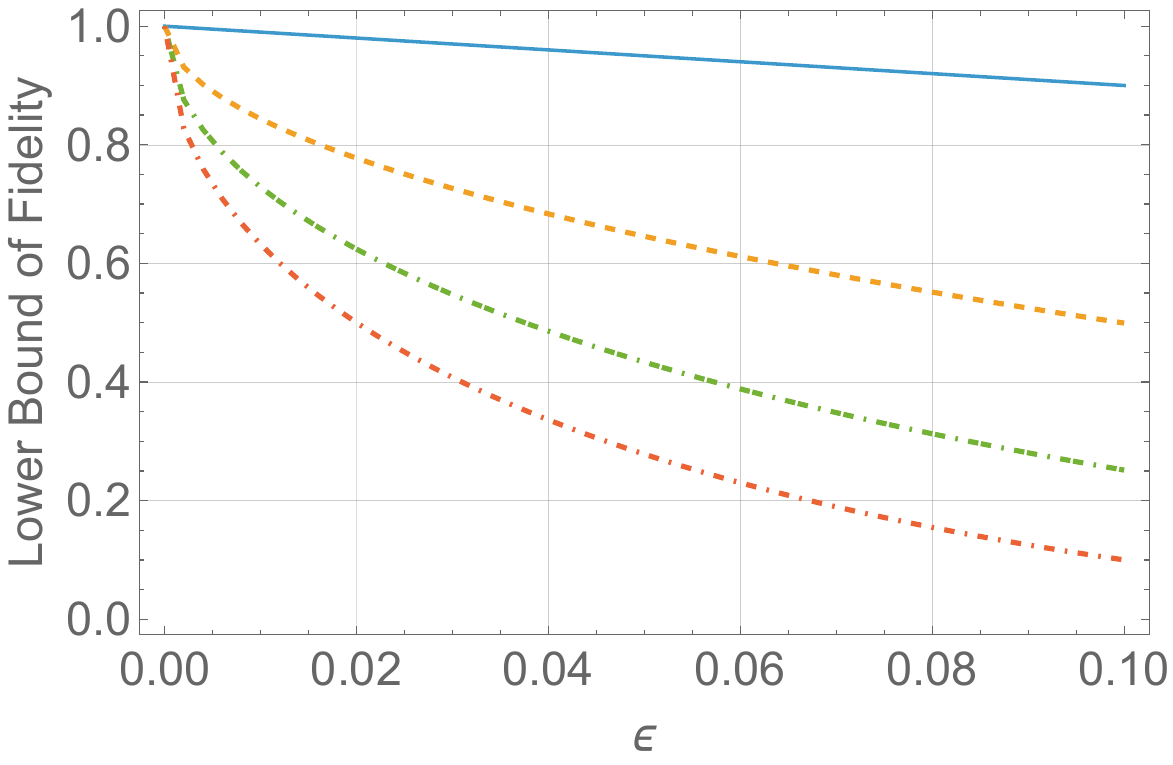}\label{fig:fid_pertur}}
   \subfigure[]{ \includegraphics[width=4.2cm]{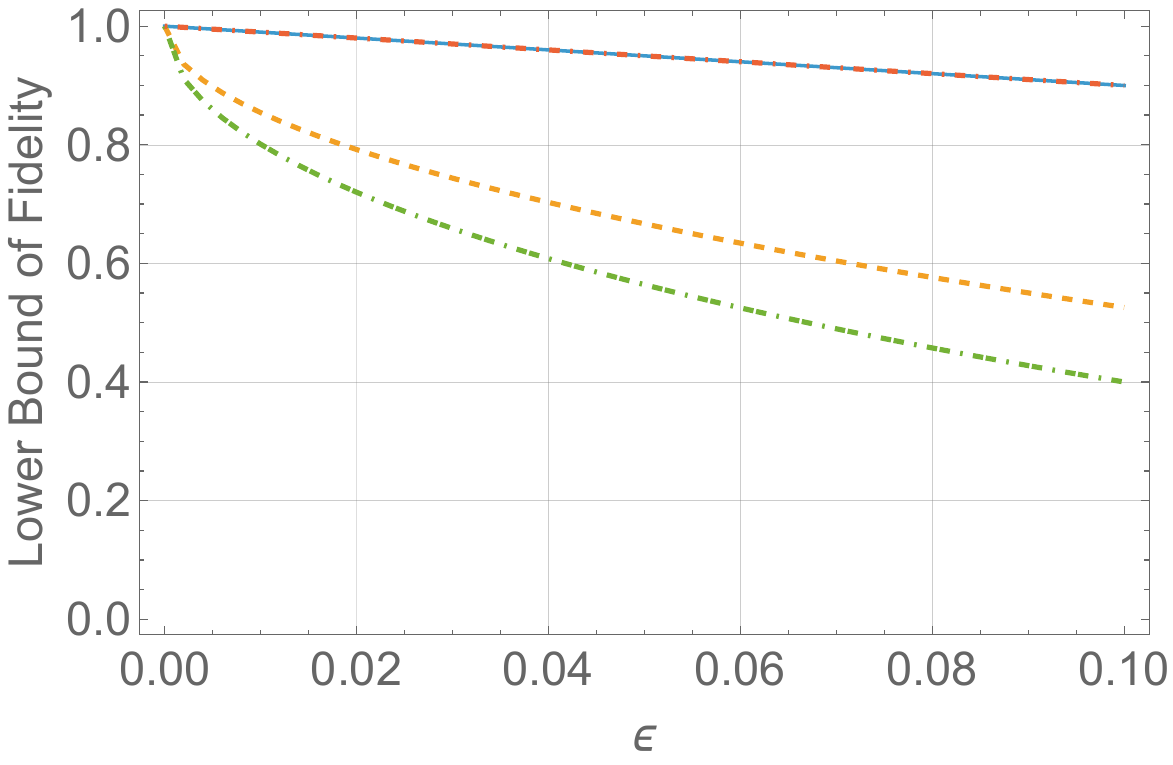}\label{fig:fid_pertur_m}}
   
\caption{SDP-based calculation of the fidelity lower bound. (a) and (b) are the fidelity lower bound for $\cos(\theta/2)|\Psi_0\rangle+e^{i\phi}\sin(\theta/2)|\Psi_1\rangle$ state, under the constraints $F_0=F_1=F_1^+=0.99$ and $F_0=F_1=F_1^+=F_1^-=0.99$, respectively. (c) and (d) show the fidelity lower bound for the state $\cos(\pi/4)|\Psi_0\rangle+e^{i\phi}\sin(\pi/4)|\Psi_1\rangle$ as a function of the perturbation $\epsilon$, given as the infidelity $F=1-\epsilon$. (c) uses constraints $F_0=F_1=F_1^+= 1-\epsilon$, based on three initial states: $|\Psi_0\rangle, |\Psi_1\rangle$, and $|\Psi_1^+\rangle$. (d) includes an additional fidelity constraint $F_1^-= 1-\epsilon$ for the initial state $|\Psi_1^-\rangle$.}
\label{fig:fid_lb_all}
\end{figure}

Figure~\ref{fig:fids_lb} illustrates a simple example of SDP calculation of the lower bounds for subspace of dimension 2, spanned by $\{|\Psi_0\rangle,|\Psi_1\rangle\}$, for given values of $F_0$, $F_1$, $F_1^+$. In the figure, the fidelity lower bound decreases as the phase difference between the examined state and the state $|\Psi_1^+\rangle$---which provides the constraint---increases, reaching a minimum at $\phi = \pi$. It is worth noting that this lower bound is derived from the minimal information ${F_0, F_1, F_1^+}$. In practice, a trained PQC may exhibit higher fidelities across various states, and the quality of the lower bound can be improved by incorporating additional informations---specifically, fidelities of other states---obtained from the trained PQC.
If one uses $|\Psi_i^{\bm{a}}\rangle \coloneqq a_0|\Psi_0\rangle + a_1|\Psi_i\rangle$ as an additional state for each $i$, the fidelity lower bound can be recomputed using Eq.~\eqref{eq:sdr} with the following additional constraint:
\begin{equation}
F_\alpha^{\bm{a}}\leq \text{Tr}\left[FC_\alpha^{\bm{a}}\right]\leq 1,\;\;\alpha = 1,2,\cdots,d-1,
\end{equation}
where the matrix $C_\alpha^{\bm{a}}$ is
\begin{equation}
[C_\alpha^{\bm{a}}]_{ij,kq} =
\begin{cases}
a^*_ia_ja_ka^*_q & \text{if } ij, kq \in \{00, 0\alpha, \alpha0, \alpha\alpha\} \\
0 & \text{otherwise}
\end{cases}.
\end{equation}

Since a state $|\Psi_1^-\rangle\coloneqq(|\Psi_0\rangle-|\Psi_1\rangle)/\sqrt{2}$---which corresponds to the case $a_0=-a_1=1/\sqrt{2}$--- yields the lowest fidelity lower bound under the minimal information, we consider the case that includes $|\Psi_1^-\rangle$ for an additional constraint, as an example. Figure~\ref{fig:fids_lb_m} presents the same calculation as in Figure~\ref{fig:fids_lb}, but with the additional constraint $F_1^-=0.99$. The figure shows that, in addition to raising the fidelity lower bound of $\phi=\pi$ (corresponding to $|\Psi_1^-\rangle)$ by using it as a constraint, those for other states also improve due to the added training information. 

In addition, we examine how the fidelity lower bound declines with increasing training imperfections, and how an additional constraint affects this behavior. 
Figure~\ref{fig:fid_pertur} and ~\ref{fig:fid_pertur_m} show the decrease in the fidelity lower bound as a function of the perturbation $\epsilon$, where $F_i=F_i^{\pm}=1-\epsilon$. Two cases are compared: one imposing the minimal constraints, and another that includes the additional constraint from the state $|\Psi_1^-\rangle$. One can see in the figure that the fidelity lower bound drops sharply for states out of phase with the states providing the constraints, as infidelity grows, but the additional constraint mitigates the overall decline. This suggests that the fidelity lower bound can be improved by incorporating the fidelities of appropriately chosen states---depending on the target states to simulate---extracted from the trained PQC, thereby providing a performance guarantee for simulating the target subspace using the PQC.

\subsection{Barren-plateau-free regions in the training landscape}
\label{sec:ws}

In the framework of VQAs, a major challenge is the barren plateau phenomenon, where the cost function exhibits exponentially vanishing gradients as the number of qubits increases, leading to exponential sampling overhead during optimization. Despite this challenge, it has been shown that, for fidelity-based cost functions, the variance of the cost function admits a plynomial lower bound in the number of qubits near certain initialization points, known as \emph{warm starts} \cite{puig2025variational}. Notably, this result has been derived specifically for the case of single-state optimization, where a PQC is iteratively trained to approximate the Trotter evolution of a single input state. This finding provides theoretical support for the effectiveness of warm-start strategies in mitigating barren plateaus during iterative training, at least in the single-state setting.

In our framework, however, the cost function is defined as the sum of multiple fidelity terms as given in Eq.~\eqref{eq:cost_m}. We, therefore, generalize the warm-start analysis for fidelity-based cost functions to the case involving a sum over multiple fidelities.
The key idea behind the generalization is that, although the fidelity terms appear mixed together within a square during the computation of the variance, the variance can be linearly decomposed into individual terms corresponding to each fidelity. As a result, for sufficiently small $\delta t$, the variance of the the cost function in Eq.~\eqref{eq:cost_m} at each time step $m$ has a polynomial lower bound with respect to the number of qubits, within a small hypercube centered at the previous optimum point from the $(m{-}1)$-th step. The formal statement is presented in Appendix~\ref{sec:ws_statement}, and the full proof is provided in Supplementary material.

\section{Results}

\subsection{Demonstration on the quantum processor: Two-qubit Ising model}
\label{sec:Ising}

We experimentally demonstrate the algorithm, executing it on the quantum processor. In the experiment, the algorithm captures the dynamics of multiple states within a subspace for a transverse- and longitudinal-field Ising model, described by the Hamiltonian: 
\begin{equation}\label{eq:ham_Ising}
H= -J\sum_{j=1}^{N-1} \sigma_j^x\sigma_{j+1}^x -g\sum_{j=1}^N\sigma_j^z -h\sum_{j=1}^N\sigma_j^x,
\end{equation}
where $N$ is the number of spins (i.e., qubits), and $J$ is the nearest-neighbor spin-spin coupling strength. The parameters $g$ and $h$ represent the strengths of the external magnetic field in the transverse and longitudinal directions, respectively. Here, we set $J = 1$, so all energy scales are measured in units of $J$.

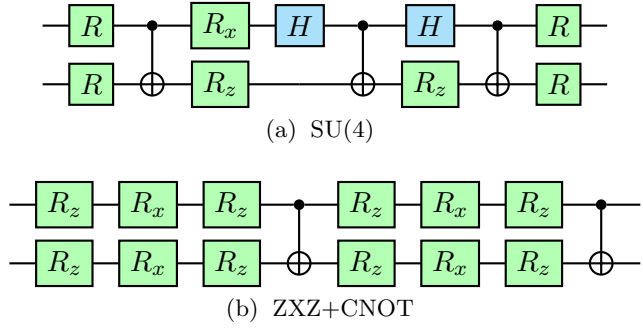
\begin{figure}[h]
\begin{center}
\subfigure[\,SU(4)]{ \begin{quantikz}[column sep =10, row sep = 5]
    \qw&\gate[style={fill=green!30}]{R}&\ctrl{1}&\gate[style={fill=green!30}]{R_x}&\gate[style={fill=cyan!30}]{H}&\ctrl{1}&\gate[style={fill=cyan!30}]{H}&\ctrl{1}&\gate[style={fill=green!30}]{R}&\qw\\
    \qw& \gate[style={fill=green!30}]{R}&\targ{}&\gate[style={fill=green!30}]{R_z}&\qw&\targ{}&\gate[style={fill=green!30}]{R_z}&\targ{}&\gate[style={fill=green!30}]{R}&\qw
    \end{quantikz}\label{fig:ans_su4}}
 \subfigure[\,ZXZ+CNOT]{ \begin{quantikz}[column sep =10, row sep = 5]
    \qw&\gate[style={fill=green!30}]{R_z}&\gate[style={fill=green!30}]{R_x}&\gate[style={fill=green!30}]{R_z}&\ctrl{1}&\gate[style={fill=green!30}]{R_z}&\gate[style={fill=green!30}]{R_x}&\gate[style={fill=green!30}]{R_z}&\ctrl{1}&\qw\\
     \qw&\gate[style={fill=green!30}]{R_z}&\gate[style={fill=green!30}]{R_x}&\gate[style={fill=green!30}]{R_z}&\targ{}&\gate[style={fill=green!30}]{R_z}&\gate[style={fill=green!30}]{R_x}&\gate[style={fill=green!30}]{R_z}&\targ{}&\qw
    \end{quantikz}\label{fig:ans_zxz}}
         \end{center}
\caption{\label{fig:ans_2q}PQCs used in the experiment. The $R$ gate denotes an Euler rotation with three parameters, and the $R_i$ gate is a single-qubit rotation about the $i$-axis. (a) An ansatz called SU(4), which represents all possible $4 \times 4$ unitary matrices up to a global phase, with 15 variational parameters. (b) A double-layered ansatz consisting of single-qubit rotations along the $z$, $x$, and $z$ axes, followed by CNOT gates, referred to as the $ZXZ$+CNOT ansatz.}
\end{figure}

The model we examine has $N = 2$ with $g = h = 1$. The target subspace for simulation is chosen to be spanned by ${|00\rangle, |11\rangle}$; accordingly, three states $|00\rangle$, $|11\rangle$, and $|+\rangle \coloneqq (|00\rangle + |11\rangle)/\sqrt{2}$ are selected as initial states.
We train two types of PQCs---the SU(4) and ZXZ+CNOT ansatz---as presented in Figure~\ref{fig:ans_2q}. For training, we use sequential minimal optimization \cite{nakanishi2020sequential}, with a halting condition of $10^{-3}$ to minimize the cost function in Eq.~\eqref{eq:cost_m} at each time step $m$.






\begin{figure}[h]

        \centering

    \subfigure[]{\includegraphics[width=4.25cm]{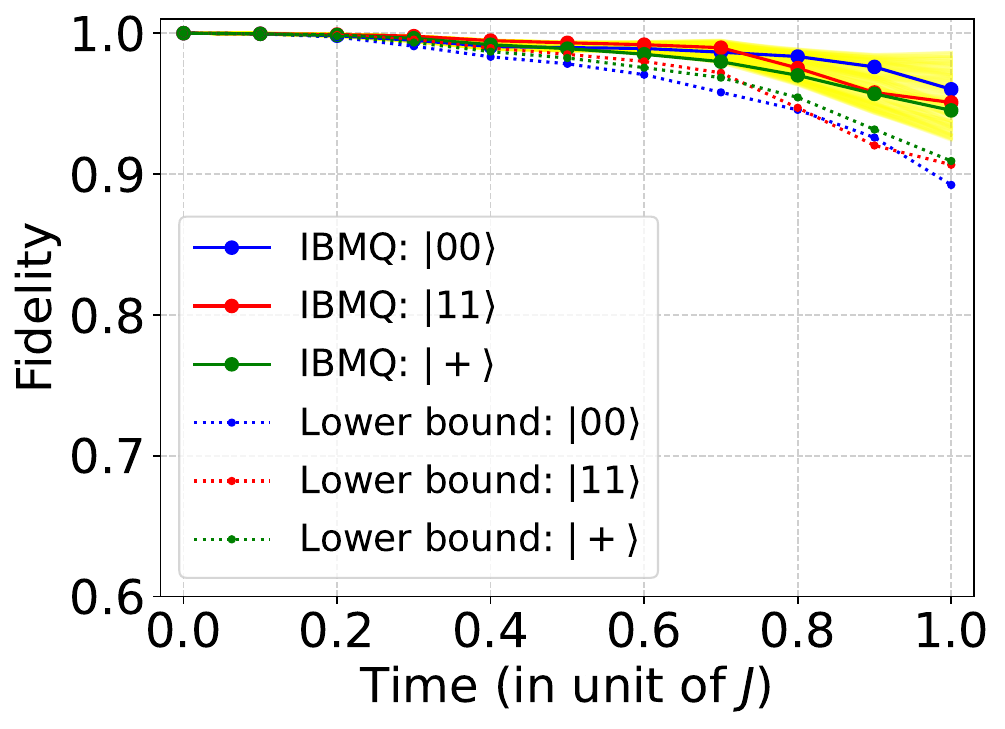}\label{fig:2q_su4}}
    \subfigure[]{\includegraphics[width=4.25cm]{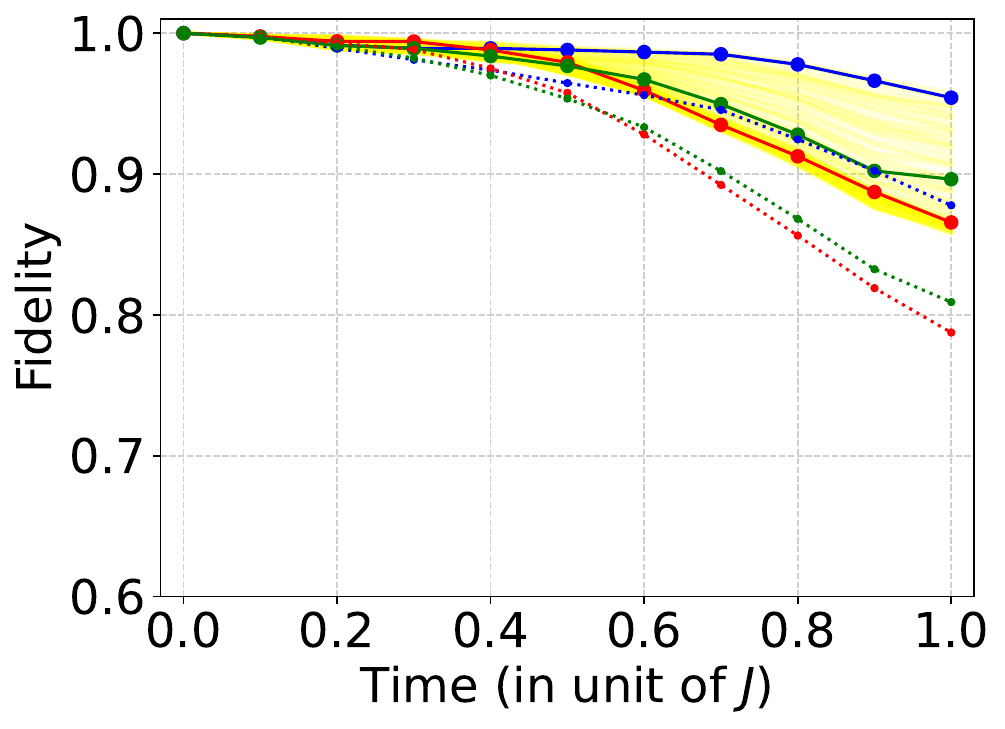}\label{fig:2q_zxz}}
\caption{\label{fig:2q_exp} Fidelities between the states evolved under the Trotter circuit and the parameterized states produced by the trained PQC. Solid lines show the fidelities of the three training target states, while dotted lines indicate their corresponding lower bounds from Proposition~\ref{prop:lb_m}. Transparent yellow lines represent the fidelities of 500 randomly sampled initial states within the subspace. (a) shows the results with the SU(4) ansatz, and (b) those with the ZXZ+CNOT ansatz. }

\end{figure}

The results of the experiment are presented in Figure~\ref{fig:2q_exp}, which demonstrates how well the trained PQCs capture the Trotter time evolution within the subspace, as measured by the fidelity between $T(\delta t)^m|\Psi\rangle$ and $|\Psi(\bar{\bm{\phi}}_m)\rangle$ at each time step $m$. As shown in the figures, the PQCs trained on the quantum processor successfully capture the time evolution of the three initial states, with the lowest fidelity remaining above 0.94 and 0.86 for the SU(4) and ZXZ+CNOT ansatz, respectively. The main interest of our study is to reproduce the dynamics of arbitrary states within a given subspace. In this context, we also examine the fidelities of 500 randomly sampled states within the subspace spanned by $\{|00\rangle, |11\rangle\}$, shown as transparent yellow lines. These fidelities are numerically computed using the optimized parameters obtained from the experiment. Notably, the dynamics of random states are well reproduced—comparable to those of the three representative states---even though their fidelity lower bounds can be significantly smaller, indicating that the PQCs are effectively trained. 


\begin{figure*}[t]
\begin{center}
\subfigure[]{ \includegraphics[width=5.5cm]{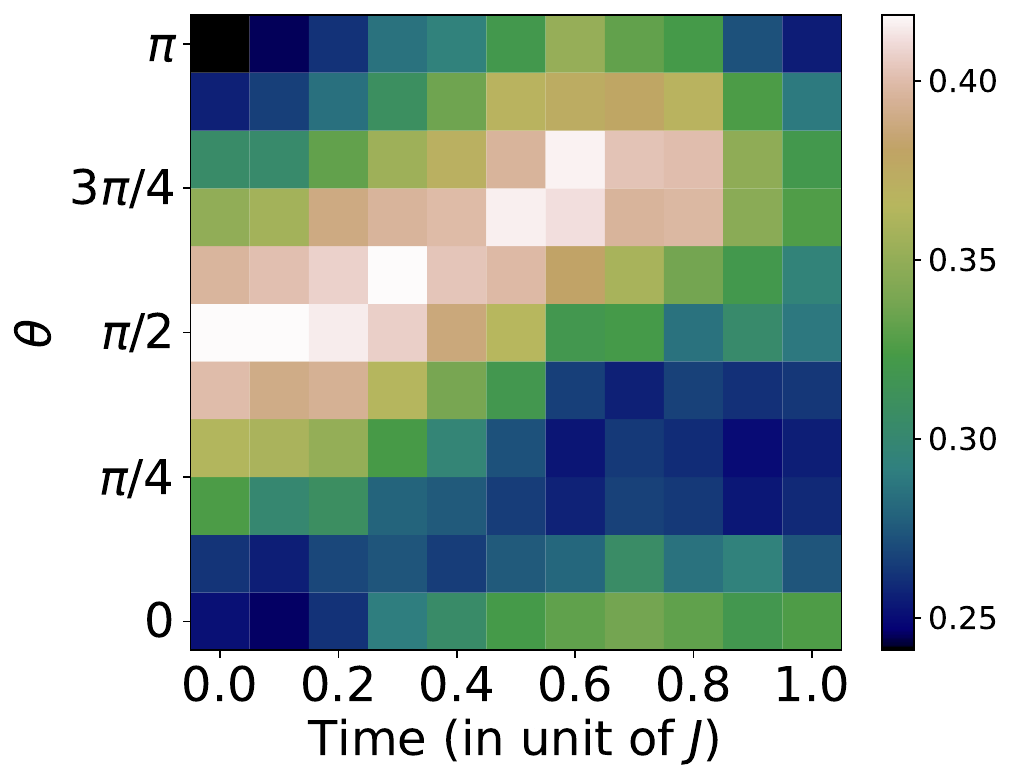}\label{fig:ent_pqc}}
 \subfigure[]{ \includegraphics[width=5.5cm]{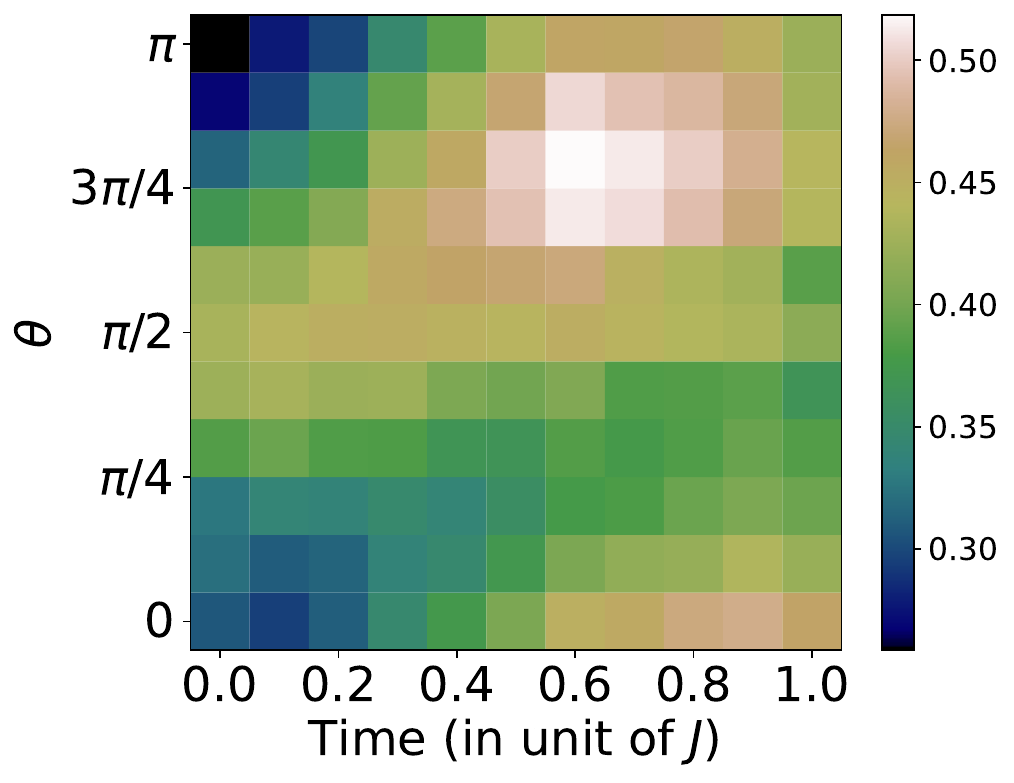}\label{fig:ent_trott}}
  \subfigure[]{ \includegraphics[width=5.5cm]{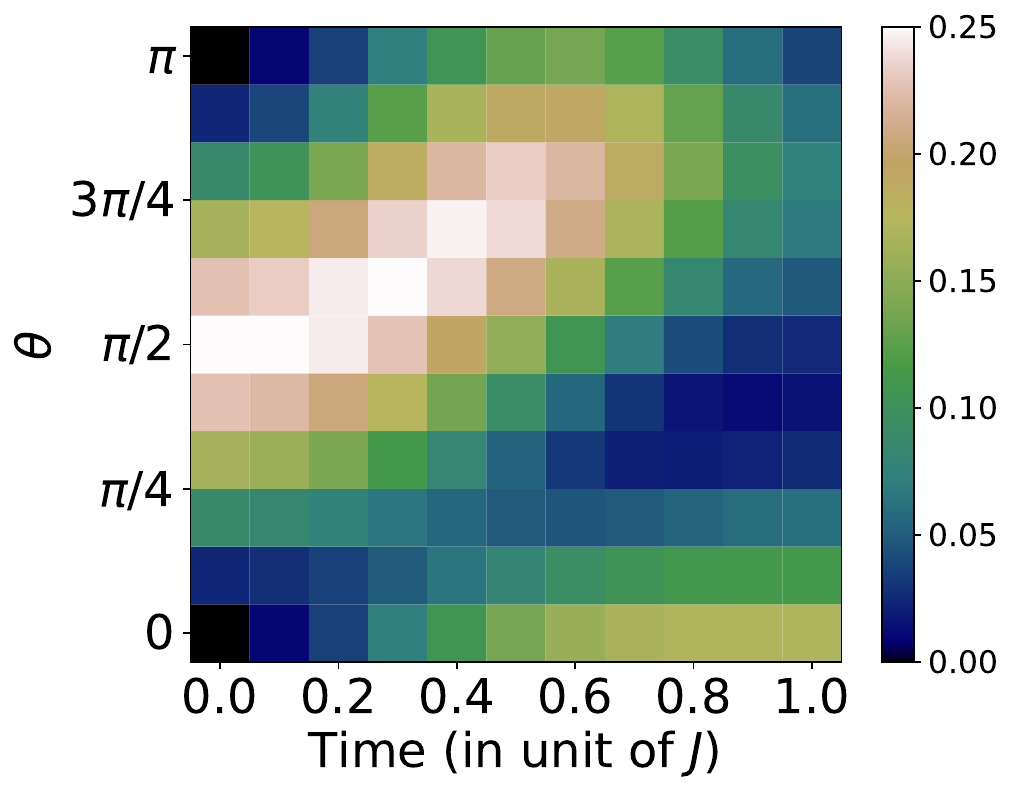}\label{fig:ent_numeric}}
         \end{center}
\caption{\label{fig:ent}Quantum simulation of the dynamics of concentratable entanglement, for multiple initial states $|\psi_0(\theta)\rangle$. Time evolution is implemented using (a) a pre-trained PQC with a SU(4) ansatz, (b) a second-order Trotter circuit, and (c) the ideal time evolution operator. Simulations in (a) and (b) are carried out on the quantum processor, while (c) is obtained by numerical calculation.}
\end{figure*}

In addition, using the trained PQC, we perform a quantum simulation of the entanglement dynamics of multiple states as an application of the algorithm. Specifically, we simulate the dynamics of concentratable entanglement \cite{beckey2021computable}---defined in a form amenable to circuit-based measurement---on the quantum processor using a 6-qubit circuit. The trained SU(4) ansatz is used to simulate the entanglement dynamics, and for comparison, we also simulate the same dynamics using the Trotter circuit. In this simulation, the concentratable entanglement is measured over time for the initial state $|\psi_0(\theta)\rangle \coloneqq\cos(\theta/2)|00\rangle + \sin(\theta/2)|11\rangle$.

The results of the entanglement dynamics simulation are presented in Figure~\ref{fig:ent}. One can see that the result in Figure~\ref{fig:ent_trott} is of relatively lower quality compared to that in Figure~\ref{fig:ent_pqc}, aside from a constant bias, with the quality gap widening over time, due to the increasing circuit depth required for Trotter time evolution. In contrast, the simulation using the PQC does not suffer from such depth-related degradation and maintains consistent performance across all time steps. Note that the PQC used to simulate time evolution of $|\psi_0(\theta)\rangle$ was trained on only three states: $|00\rangle,|11\rangle,|+\rangle$ (corresponding to $\theta = 0, \frac{\pi}{2}, \pi$). Even though the time evolution of the entire state set $\{\psi_0(\theta)\rangle\}$ is reproduced without separate optimization for each state, it still closely follows the target evolution (as also seen in Figure~\ref{fig:2q_su4}), thus successfully reproduce the entanglement dynamics over the multiple initial states.

\subsection{Simulation: 10-qubit Ising model}
\label{sec:10Ising}
For further validation of our algorithm, we perform a numerical simulation of the algorithm on the 10-qubit transverse- and longitudinal-field Ising model, which corresponds to the $N=10$ case of Eq.~\eqref{eq:ham_Ising}. As in Section~\ref{sec:Ising}, setting $J=1$ thus all energy scales are dimensionless in units of $J$, and $g=h=1$. 

To approximate the Trotter circuit for the $10$-qubit Ising model, we use an ansatz that sequentially stacks 2-qubit SU(4) PQCs across all $10$ qubits, as depicted in Figure~\ref{fig:ans_10q}. Defining a layer as one cycle in which nine SU(4) blocks sequentially cover the entire qubit wire, we simulate the algorithm using both single- and double-layer configurations. The PQCs are trained for the Trotter time evolution of the subspace spanned by $\{|00\cdots 0\rangle, |11\cdots 1\rangle\}$, thus the initial states are given by $|00\cdots 0\rangle$, $|11\cdots 1\rangle$ and $\frac{1}{\sqrt{2}}(|00\cdots 0\rangle+|11\cdots 1\rangle)$. For training, we use the stochastic gradient descent (SGD) optimizer with learning rate $0.1$, implemented in the Optax library \cite{deepmind2020jax}. 

\begin{figure}[h]
\begin{center}
\begin{quantikz}[column sep =10, row sep = 1]
    \qw&\gate[2,style={inner ysep=-6pt, fill=green!30}]{\text{SU}(4)}&\qw&\qw&\qw&\qw&\qw\\
    \qw&&\gate[2,style={inner ysep=-6pt, fill=green!30}]{\text{SU}(4)}&\qw&\qw&\qw&\qw\\\qw&\qw&&\gate[2,style={inner ysep=-6pt, fill=green!30}]{\text{SU}(4)}&\qw&\qw&\qw\\\qw&\qw&\qw&&\qw&\qw&\qw\\&&&&\ddots&&
    \end{quantikz}
         \end{center}
\caption{\label{fig:ans_10q}The single-layer ansatz used to capture the Trotter time evolution of the 10-qubit Ising system. Each 2-qubit SU(4) block is shown in Figure~\ref{fig:ans_su4}, thus, a total of 15$\times (n-1)$ parameters are encoded in the ansatz per a layer for an $n$-qubit system. }
\end{figure}
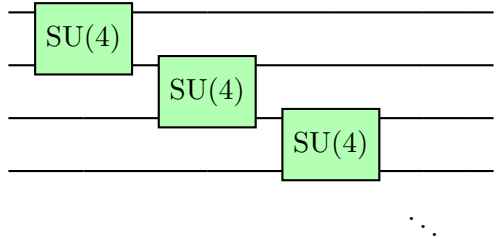

\begin{figure}[h]
\begin{center}
\subfigure[]{ \includegraphics[width=4.1cm]{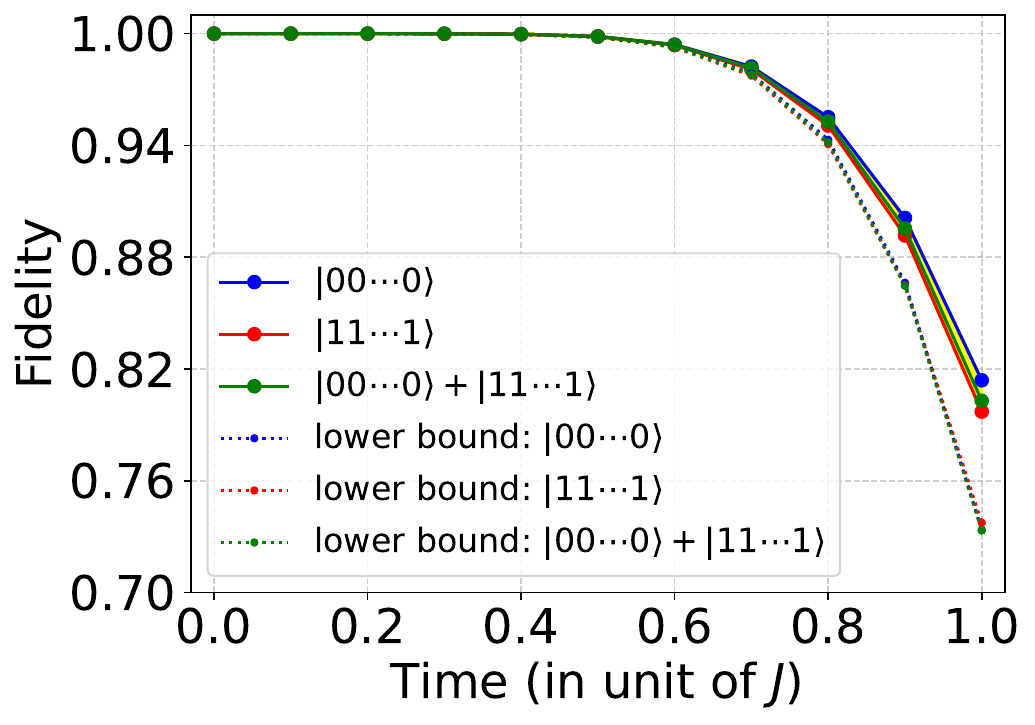}}\label{fig:10_q_Ising_layer_1}
 \subfigure[]{ \includegraphics[width=4.2cm]{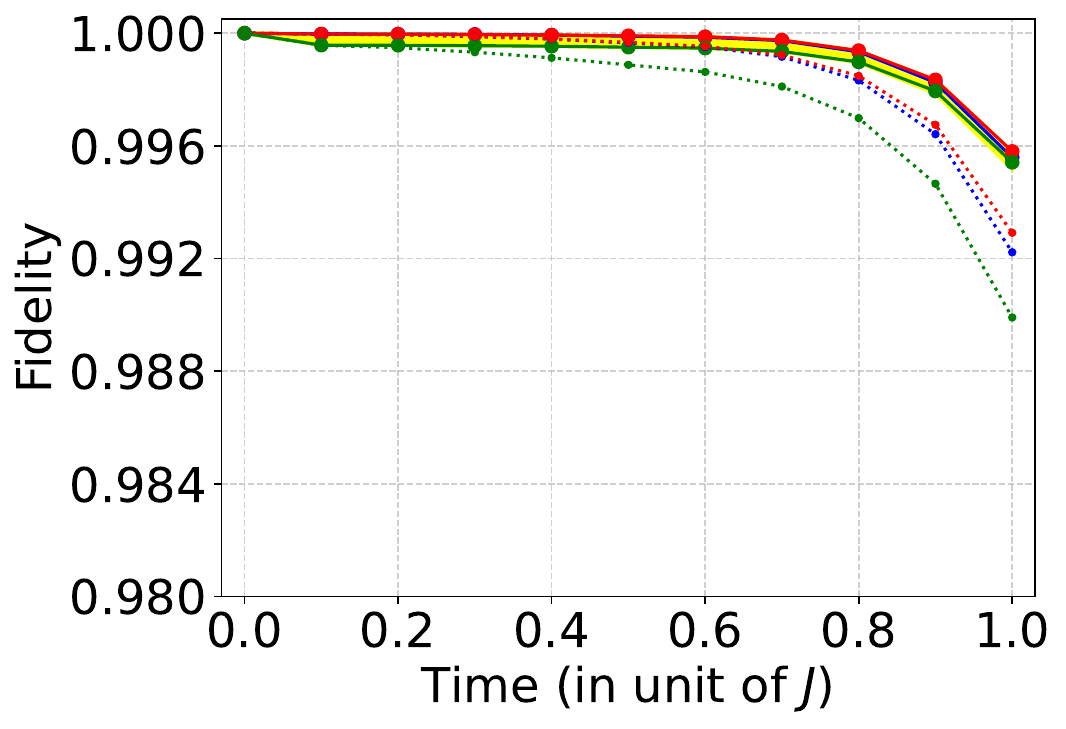}}\label{fig:10_q_Ising_layer_2}
         \end{center}
\caption{Fidelities---and their lower bounds as presented in Proposition~\ref{prop:lb_m}---between the states evolved under the Trotter circuit and those approximated by PQCs, for three representative initial states. Additionally, fidelities for 500 random initial states within the subspace spanned by $\{|00\cdots0\rangle, |11\cdots1\rangle\}$ are shown as transparent yellow lines. (a) Using a single-layer ansatz (b) Using a double-layer anstaz.}
\label{fig:10_q_Ising}
\end{figure}

The results of the simulation are presented in Figure~\ref{fig:10_q_Ising}, demonstrating that the algorithm is capable of capturing the dynamics of arbitrary states within the subspace spanned by $\{|00\cdots0\rangle, |11\cdots1\rangle\}$, even for large systems, as shown by the fidelities. We first note that the single-layer PQC exhibits a noticeable drop in fidelity over time, whereas the double-layer PQC maintains an error below 1\% throughout the later time steps. This indicates that a deeper circuit depth is required to accurately simulate quantum systems over longer evolution times, resembling the general requirement that circuit depth grows linearly with simulation time, in product-formula-based quantum simulation \cite{childs2019nearly, haah2021quantum, barison2021efficient}. In both cases, the trained PQCs successfully reproduce the dynamics of arbitrary states within the subspace, as well as those of the three representative states directly used as training targets. This example shows that the algorithm can accurately simulate Trotter time evolution across multiple states in a given subspace, for larger systems, provided that the ansatz has sufficient expressibility.

\section{Discussion}\label{sec:discussion}

In this paper, we propose a method to simulate the Trotter time evolution of a subspace using PQCs. The subspace can be arbitrarily chosen by preparing a set of orthonormal initial states that serve as its basis. After training the PQC with the $2d-1$ initial states to capture the $d$ columns of the matrix representation of the Trotter circuit (restricted to the subspace) in the given basis, the PQC can reproduce the dynamics of any linear combination of the basis states. Notably, a single training procedure suffices to reproduce the dynamics of arbitrary states within the subspace, whereas algorithms for single state require retraining for each new initial state. Furthermore, the training landscape of the algorithm provides warm-start---barren-plateau-free---regions for each iteration. While this warm-start property has been proven for iterative algorithms restricted to the single-state training \cite{puig2025variational}, we extend the theorem to our algorithm, which addresses the multi-state optimization setting. Consequently, our algorithm maintains sufficient trainability for large systems, whereas VQAs often fail due to gradients that vanish exponentially with system size.

The success of our algorithm is measured by fidelity. In this context, the fidelity lower bound for states within the given subspace provides a performance guarantee for the worst-case training scenario of the PQC. While computing this bound is generally NP-hard in the framework of QCQP, we relax the problem and introduce an SDP-based method for efficiently calculating the fidelity lower bound between arbitrary states (within the subspace) evolved under the Trotter circuit and those produced by the trained PQC. This bound is evaluated directly from the cost function values naturally obtained during the training process.

The algorithm is experimentally demonstrated by simulating the dynamics of $2$-qubit transverse- and longitudinal-field Ising model on the Eagel processor \texttt{ibm\_yonsei}. The PQC trained using the QPU successfully reproduces the time evolution within the two-dimensional subspace, achieving fidelities above 0.9 for arbitrary states when using an ansatz with sufficient expressibility, despite the presence of noise on the device. Furthermore, using the trained PQC, we experimentally investigate the dynamics of entanglement within the subspace on the same QPU. Although the effect of noise is more pronounced due to the use of six qubits, the PQC captures the overall trend of the entanglement dynamics. We additionally provide numerical simulation results by executing the algorithm on the 10-qubit Ising model, demonstrating its capability for simulating large-scale quantum systems.

Since the fidelity lower bound can be computed instantaneously during the training process, one can monitor whether the PQC reproduces the target dynamics with the desired accuracy at each iteration. A natural strategy is to increase the expressibility of the PQC by adding additional layers when the fidelity lower bound falls below a prescribed accuracy threshold. Such a modification can affect certain properties of the algorithm (e.g., the warm-start property) and therefore requires further careful investigation.

Beyond simulating the dynamics of quantum states, the algorithm can serve as a versatile subroutine to improve the performance of other computational methods. In particular, it can facilitate the preparation of multiple Krylov basis states \cite{stair2020multireference,oliveira2025quantum} by exploiting its expressibility and trainability (as discussed in Section~\ref{sec:ws}) within the constraints of near-term quantum devices. These features may, in turn, enable a variety of applications, including solving eigensolver problems \cite{cortes2022fast} and simulating long-time quantum dynamics \cite{yoshioka2025krylov}.

\appendix
\section{Proofs of Section~\ref{sec:lb}}\label{sec:proof_lb}

\subsection{Proof of Proposition~\ref{prop:lb_m}}
\begin{proof}
While the fidelity between any two states $|\psi\rangle$, $|\psi'\rangle$ is not a metric, the angle, $\Theta(|\psi\rangle,|\psi'\rangle)\coloneqq \arccos |\langle\psi|\psi'\rangle|$ obeys the triangle inequality \cite{Nielsen_Chuang_2010}. By using the triangle inequality, it is straightforward to obtain
\begin{widetext}
\begin{align}
    \Theta\left(|\Psi_i(\bar{\bm{\phi}}_m)\rangle,T(\delta t)^m|\Psi_i\rangle\right) \leq \Theta\left(|\Psi_i(\bar{\bm{\phi}}_m)\rangle,T(\delta t)^{m-1}|\Psi_i(\bar{\bm{\phi}}_1)\rangle\right)+\Theta\left(|\Psi_i(\bar{\bm{\phi}}_1)\rangle,T(\delta t)|\Psi_i\rangle\right),   
   \end{align}
\end{widetext}
for all $i$. Repeating this sequentially, we obtain the relation Eq.~\eqref{eq:lowerbound_m}.
\end{proof}

\subsection{Proof of Theorem~\ref{thm:lb}}\label{sec:pf_sd}
\begin{proof}
By the primal-dual interior-point path following method \cite{alizadeh1998primal}, the problem Eq.~\eqref{eq:sdr} can be solved in $O(d^7)$ time per iteration, or more efficiently with faster algorithms \cite{jiang2020faster}.

The strong duality is established by constructing a feasible $F$ under the condition $F_\alpha, F_\beta^+ < 1$ for all $\alpha, \beta$ as follows:

\begin{align}
[F]_{ij,kq} = \begin{cases}
F_{\text{max}} & \text{if}\; i = j = k = q \\
F_{\text{max}}-\epsilon & \text{if}\; i=j,\;k=q\;\text{and} \;i\neq k\\
\epsilon & \text{if}\; i=k,\;j=q\;\text{and}\;i\neq j\\
0 &\text{otherwise}
\end{cases} ,\nonumber\\\left(F_{\text{max}}=\text{max}(\{F_\alpha\},\{F_\beta^+\})\;\;,\;\;\;0<\epsilon<\frac{1-F_{\text{max}}}{d-1}\right).
\end{align}

This $F$ is strictly feasible---i.e., it is positive definite and satisfies all constraints---thus satisfying Slater’s condition. Consequently, the semidefinite program admits strong duality.
\end{proof}

\section{Comparison with the fewer state optimization}
\label{sec:comp_fewer}

The key difference between our algorithm and previously studied iterative algorithms for single-state simulation \cite{berthusen2022quantum, lin2021real} lies in the choice of target initial states used in the cost function Eq.~\eqref{eq:cost_m}. As explained in Section~\ref{sec:framework}, to train the sub-blocks (restricted to the target subspace) of the Trotter matrix, the algorithm requires not only the basis states $\{|\Psi_i\rangle\}$ but also their linear combinations $\{|\Psi_i^+\rangle\}$. Here, we examine the necessity of including these states for reproducing the dynamics within a given subspace, based on the simulation results under the same setting as in Section~\ref{sec:Ising}. We train the PQCs---SU(4) and ZXZ+CNOT ansatz---for three different sets of target states: $\{|00\rangle, |11\rangle, |+\rangle\}$, $\{|00\rangle, |11\rangle\}$, and $\{|00\rangle\}$. We again use the sequential minimal optimization \cite{nakanishi2020sequential}, but with a tighter halting condition of $10^{-4}$ for the simulation. 

\begin{figure*}[t]

     \subfigure[]{\includegraphics[width=4.3cm]{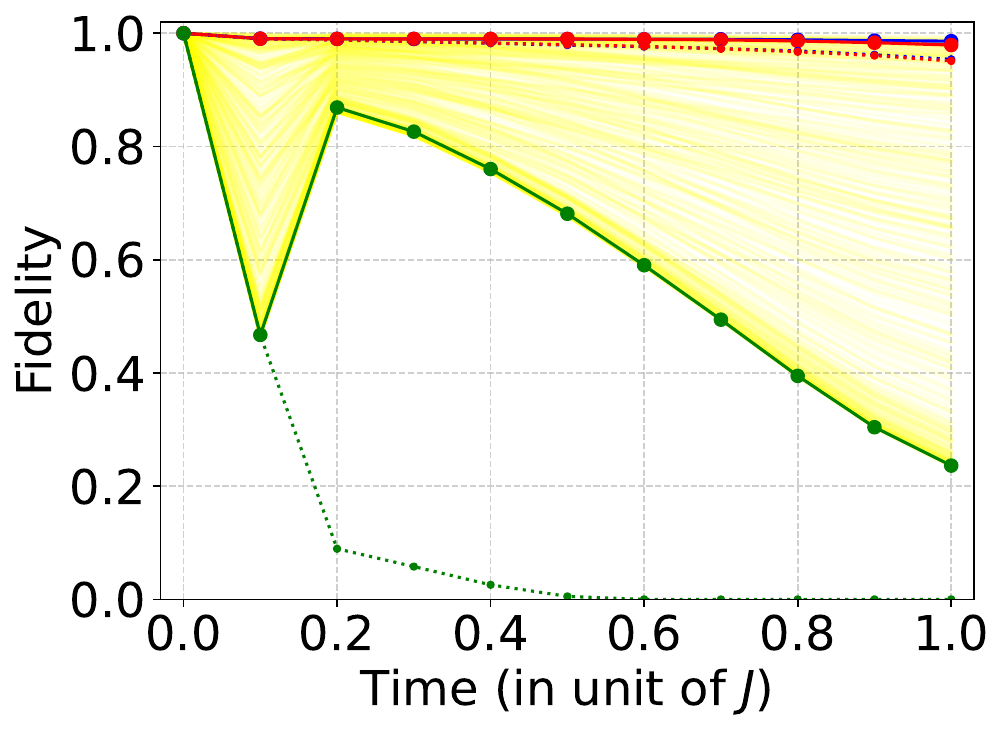}\label{fig:su4_double}}
     \subfigure[]{\includegraphics[width=4.3cm]{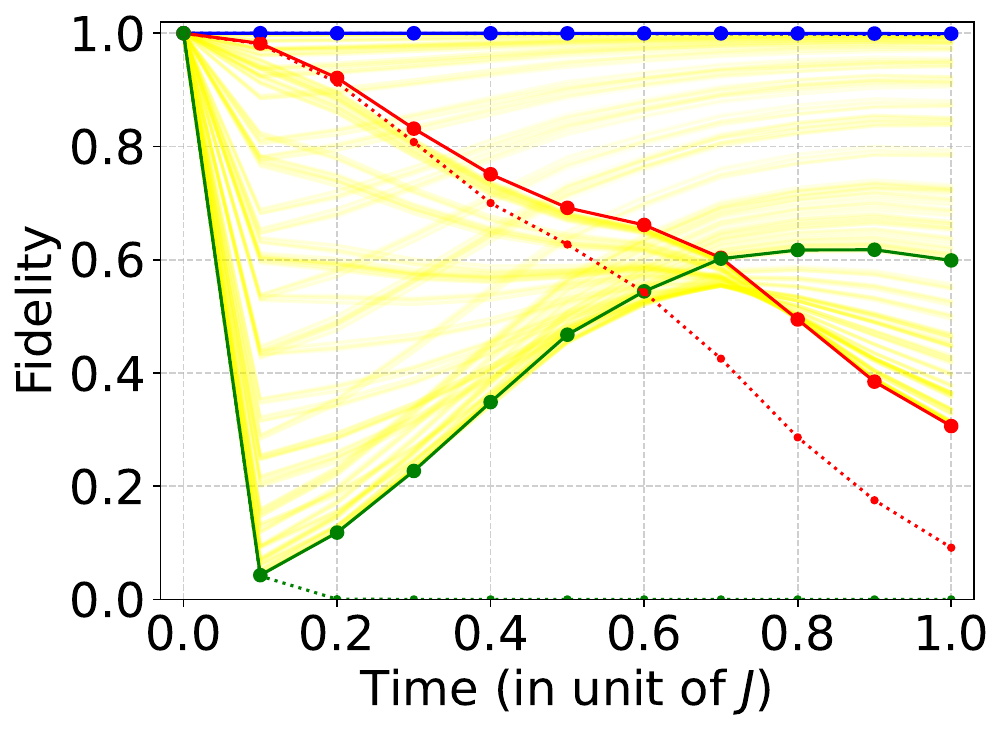}\label{fig:su4_single}}
     \subfigure[]{\includegraphics[width=4.3cm]{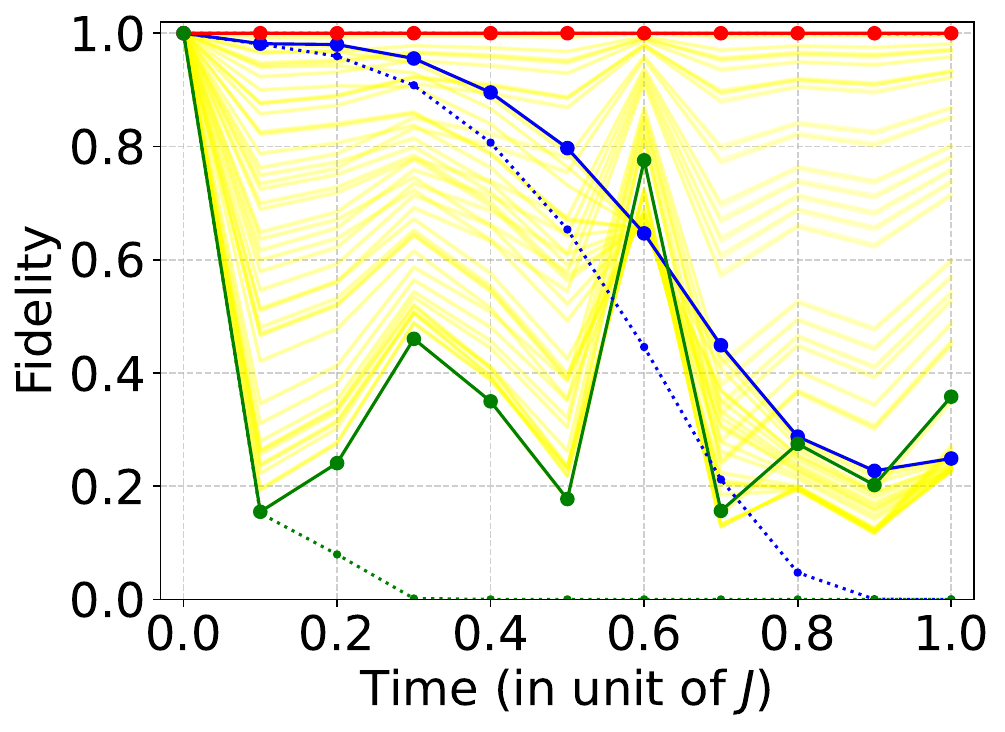}\label{fig:su4_single_11}}
  
        \subfigure[]{ \includegraphics[width=4.3cm]{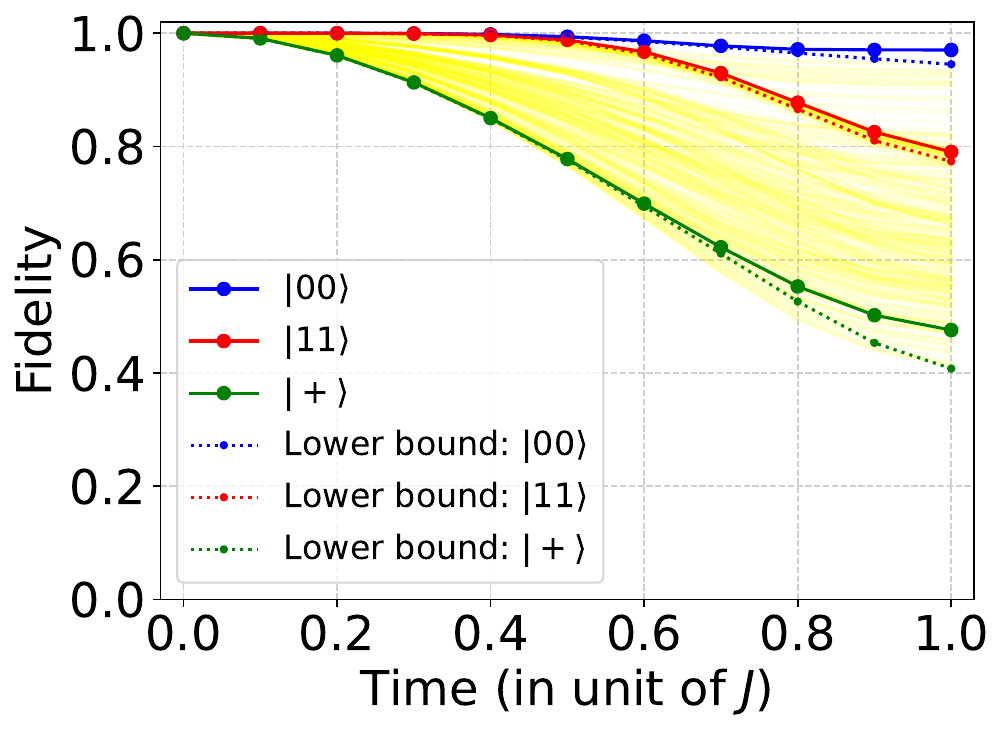}\label{fig:zxz_double}}
      \subfigure[]{ \includegraphics[width=4.3cm]{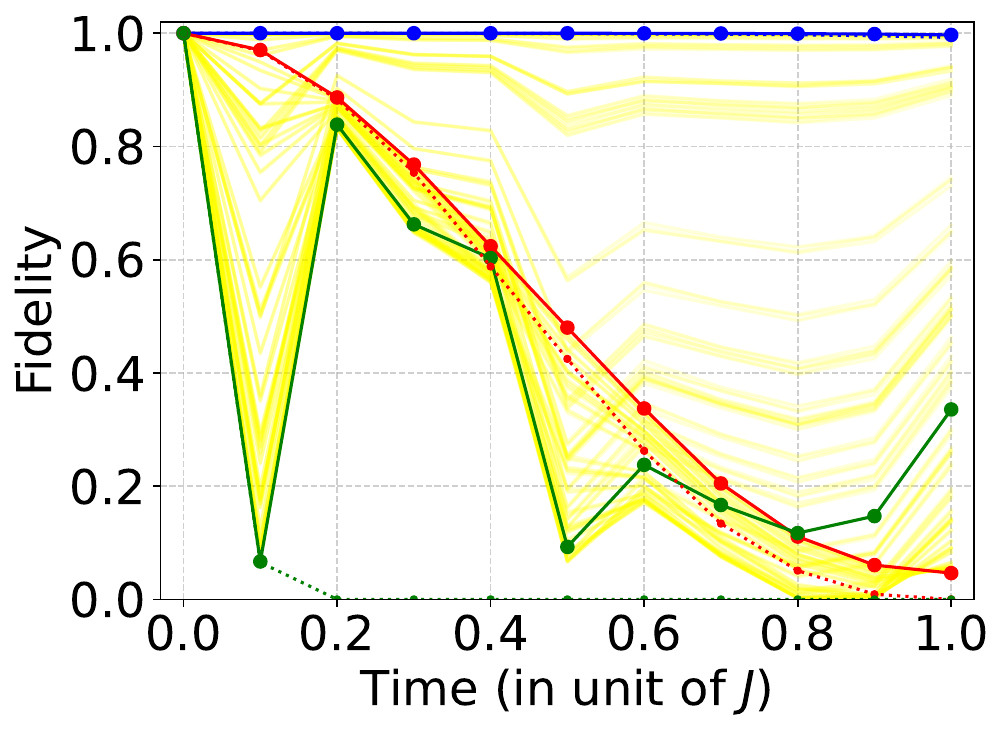}\label{fig:zxz_single}}
    \subfigure[]{\includegraphics[width=4.3cm]{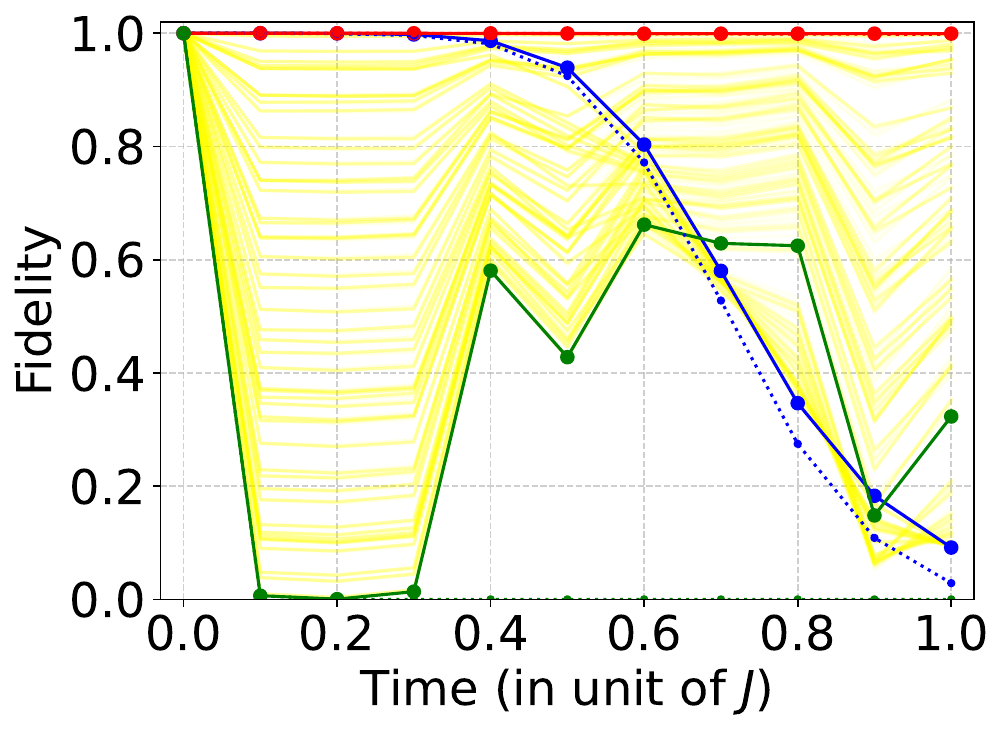}\label{fig:zxz_single_11}}
\caption{\label{fig:2q_fewer} Fidelities between the states evolved under the Trotter circuit and the parameterized states produced by the trained PQC. The fidelities of the three representative states $|00\rangle$, $|11\rangle$, and $|+\rangle$ are shown as solid lines, and their corresponding fidelity lower bounds—given in Proposition~\ref{prop:lb_m}—are shown as dotted lines. Additionally, the fidelities of 500 randomly sampled initial states within the subspace are displayed as transparent yellow lines. (a), (b) and (c) show the results using the SU(4) ansatz, trained with the sets of initial states $\{|00\rangle, |11\rangle\}$, $\{|00\rangle\}$ and $\{|11\rangle\}$ respectively. (d), (e) and (f) present the corresponding results using the ZXZ+CNOT ansatz.}
\end{figure*}

Figure~\ref{fig:2q_fewer} demonstrates the necessity of the entire target states $\{|00\rangle, |11\rangle, |+\rangle\}$. As shown in Figures~\ref{fig:su4_double} and~\ref{fig:zxz_double}, the lack of the $|+\rangle$ state leads to incorrect training of the subspace dynamics, whereas the fidelities of the two target states $|00\rangle$ and $|11\rangle$ remain high. The gap in training quality is clearly demonstrated when compared to the results in Figure~\ref{fig:2q_exp}, which exhibit significantly better performance despite being obtained on a noisy device with a weaker halting condition. As expected, training with only the single state $|00\rangle$ fails to reproduce the subspace dynamics---performing even worse than using two states---as shown in Figures~\ref{fig:su4_single} and~\ref{fig:zxz_single}.

From this example, one can observe a trade-off between trainability and the dimension of the target subspace. Note that the ZXZ+CNOT ansatz possesses sufficient expressibility to reproduce the dynamics of individual states such as $|00\rangle$ and $|11\rangle$, as shown in Figures~\ref{fig:zxz_single} and~\ref{fig:zxz_single_11}. However, as demonstrated in Figure~\ref{fig:zxz_double} (and Figure~\ref{fig:2q_zxz}), when the training involves a larger set of target states, the ansatz fails to accurately capture the dynamics across the full subspace. In contrast, the SU(4) ansatz does not exhibit this limitation, as presented in Figure~\ref{fig:su4_double} (and Figure~\ref{fig:2q_su4}). This suggests that an ansatz with sufficiently high expressibility can overcome the trade-off.

\section{Formal statements of the barren-plateau-free regions}
\label{sec:ws_statement}

Here, we present explicit statements about the warm start given in Section~\ref{sec:ws}, including all relevant mathematical conditions. In the following, we focus on the $2$-dimensional case, with the initial states $|\Psi_0\rangle,|\Psi_1\rangle$, and $|\Psi_1^+\rangle$. The extension to higher dimensional cases is straightforward by analogy. While we have utilized the Trotter approximation as a reference for time evolution for practical reasons, the algorithm itself does not rely on any specific quantum simulation circuit to generate the dynamics. In this statement, we use the exact time evolution operator $\mathcal{T}(\delta t)\coloneqq e^{-iH\delta t}$ to demonstrate the warm-start property in a more fundamental setting. The warm-start within the setting using the Trotter approximation can also be established in a similar way, by incorporating an additional analysis of the Trotter error \cite{yi2022spectral}. 

Some notations are given as follows: The cost function composed of the three states in the statement is denoted by 
\begin{equation}
C(\bm{\phi})=\frac{1}{3}(C_0(\bm{\phi})+C_1(\bm{\phi})+C_+(\bm{\phi})),
\end{equation}
where the subscript refers to the state index, not the iteration number, since the analysis is independent of the iteration step. Each term is defined as $C_j(\bm{\phi})=1-|\langle\psi_j|U^\dagger(\bm{\phi})\mathcal{T}(\delta t)U(\bar{\bm{\phi}})|\psi_j\rangle|^2$. The PQC with $M$-parameters is described by
\begin{equation}
U(\bm{\phi})=\prod_{i=1}^MW_i U_i(\phi_i)
\end{equation}
where $U_i(\phi_i)=e^{-i\phi_i\sigma_i}$, with $\sigma_i$ being a Pauli string associated with the $i$-th gate, and $W_i$ is a fixed gate. The set $\mathcal{V}$, as it appears in the statement, denotes a hypercube in parameter space, defined as $\mathcal{V}(\bm{\phi},r)=\{\bm{\theta}\,|\,\theta_i\in[\phi_i-r,\phi_i+r]\;\forall i\}$ and $\mathcal{D}(\bm{\phi},r)$ is the uniform distribution over the hypercube $\mathcal{V}(\bm{\phi},r)$. The system size is $n$ qubits.

\begin{proposition}
\label{prop:var}
The variance of the cost function $C(\bm{\phi})$ over the hypercube $\mathcal{V}(\bar{\bm{\phi}},r)$, with the pre-optimized parameters $\bar{\bm{\phi}}$, is lower bounded as 
\begin{widetext}
\begin{equation}\label{eq:lb_primal}
\textup{Var}_{\bm{\phi}\sim\mathcal{D}(\bar{\bm{\phi}},r)}[C(\bm{\phi})]\geq (c_+-k_+^2)\underset{\xi\in[-1,1]}{\textup{min}}\left(k_+^{M-1}\Delta_{\bar{\bm{\phi}}}+(1-k_+^{M-1})\xi\right)^2,
\end{equation}
where 
\begin{align}
c_+&\coloneqq\mathbb{E}_{\alpha\sim\mathcal{D}(0,r)}[\cos^4\alpha]\\k_+&\coloneqq\mathbb{E}_{\alpha\sim\mathcal{D}(0,r)}[\cos^2\alpha]\\\Delta_{\bar{\bm{\phi}}}&\coloneqq\frac{1}{3}\sum_{j\in\{0,1,+\}}\textup{Tr} [(\rho_j-\sigma_1\rho_j\sigma_1)U^\dagger(\bar{\bm{\phi}})\rho_{j,(\bar{\bm{\phi}},
\delta t)}U(\bar{\bm{\phi}})].
\end{align}
\end{widetext}
Here, the density matrices are defined as $\rho_j\coloneqq|\psi_j\rangle\langle\psi_j|$ and $\rho_{j,(\bar{\bm{\phi}},\delta t)}\coloneqq\mathcal{T}(\delta t)U(\bar{\bm{\phi}})\rho_j U^\dagger(\bar{\bm{\phi}})\mathcal{T}^\dagger(\delta t)$.

\end{proposition}

\begin{theorem}
\label{thm:var}
If the time step $\delta t$ of the time evolution per each iteration satisfies 
\begin{equation}
\delta t< \frac{\sqrt{3-\frac{2}{3}\sum_{j\in\{0,1,+\}}\textup{Tr}[\sigma_1\rho_j\sigma_1\rho_j]}-1}{2|E_{\textup{m}}|}
\end{equation}
with the $E_{\textup{m}}$ being the energy with the largest absolute value, among the eigenvalues of $H$, and the perturbation $r$ obeys 
\begin{widetext}
\begin{equation}
r^2\leq \left(\frac{3r_0^2}{M-1}\right)\left(\frac{1-\frac{1}{3}\sum_{j\in\{0,1,+\}}\textup{Tr}[\sigma_1\rho_j\sigma_1\rho_j]-2|E_{\textup{m}}|\delta t-2E_{\textup{m}}^2\delta t^2}{2-\frac{1}{3}\sum_{j\in\{0,1,+\}}\textup{Tr}[\sigma_1\rho_j\sigma_1\rho_j]-2|E_{\textup{m}}|\delta t-2E_{\textup{m}}^2\delta t^2}\right),
\end{equation}

with some $0<r_0<1$, then the variance of the cost function over the hypercube $\mathcal{V}(\bar{\bm{\phi}},r)$ is lower bounded by
\begin{equation}
\textup{Var}_{\bm{\phi}\sim\mathcal{D}(\bar{\bm{\phi}},r)}[C(\bm{\phi})]\geq \frac{4r^4}{45}\left(1-\frac{4r^2}{7}\right)\left(\left(1-r_0^2\right)\left(1-\frac{1}{3}\sum_{j\in\{0,1,+\}}\textup{Tr}[\sigma_1\rho_j\sigma_1\rho_j]-2|E_{\textup{m}}|\delta t-2E_{\text{m}}^2\delta t^2\right)\right)^2.
\end{equation}
\end{widetext}
Furthermore, if $|\langle\psi_j|\sigma_1|\psi_j\rangle|<1-\mathcal{O}(1/\textup{poly}(n))$, for some $j$, the variance is at least polynomial about $M$ and $n$, such that $\textup{Var}_{\bm{\phi}\sim\mathcal{D}(\bar{\bm{\phi}},r)}[C(\bm{\phi})]\in\Omega(1/(M^{2}\,\textup{poly}(n)))$ within the perturbation $r\in\Theta(1/\sqrt{M})$.
\end{theorem}

\section*{Data Availability}
The datasets generated and analysed during the current study are available in the repository:  https://github.com/sPark9144/svqs-fid-lb. 
\section*{Code Availability}
The underlying code for this study is accessed via this link: https://github.com/sPark9144/svqs-fid-lb.

\section*{Acknowledgement}

We acknowledge the Yonsei University Quantum Computing Project Group for providing support and access to the Quantum System One (Eagle processor). We are also grateful to the Quantum Computing Lab at the Electronics and Telecommunications Research Institute (ETRI) for valuable discussions and technical support.

\section*{Funding}

This work was supported by the National Research Foundation of Korea (NRF) grant funded by the Korea Government (MSIT) (Grant No. 2022M3E4A1077094, RS-2023-NR119931, RS-2023-00281456, RS-2024-00432214, and RS-2025-03532992), and  the Korea Institute of Science and Technology Information (KISTI) (Grant No. P25027). This research was supported by the National Research Council of Science \& Technology(NST) grant by the Korea government (MSIT) (No. GTL25011-410).

\section*{Author Contributions}
S.P. proved the main theorems, developed the code, performed the calculations, obtained the experimental results, and drafted the manuscript.
D.L. contributed to the code for the experimental results. 
J.B. contributed to the development of the conceptual design and methodology.
H.R. assisted in code development and performed calculations.
K.B. supervised the project, contributed to the conceptual design and methodological development, and co-drafted the manuscript.
All authors participated in the discussion of results and revision of the manuscript.

\section*{Competing interests}

The authors declare no competing interests.

\section*{Additional information}

Correspondence and requests for materials should be addressed to
Kyunghyun Baek.

\bibliography{svqs}

\end{document}